\documentclass[12pt]{article}

\usepackage{amsmath, amstext, amsfonts, amssymb, commath}
\usepackage[active]{srcltx}
\usepackage[english]{babel}
\usepackage[utf8]{inputenc}
\usepackage[T1]{fontenc}
\usepackage{unicode_character}
\usepackage{cite}               % citation in sorted order
\usepackage{bbm}
\usepackage[affil-it]{authblk}

\usepackage{txfonts}
\usepackage{dsfont}
\usepackage{braket,mleftright}
\usepackage{upgreek}
\usepackage{etoolbox}

\usepackage{pgfplots, graphicx, tikz}
\usetikzlibrary{arrows}
\usetikzlibrary{decorations.markings}
\usepackage{subcaption}
\usepackage[labelfont=bf]{caption}
\usepackage{algorithm}
\usepackage{algpseudocode}
\usepackage{stackengine}
\stackMath

\newcounter{sqindex}

\usepackage[colorinlistoftodos]{todonotes}
\usepackage{xcolor}

\newcommand{\Rn}{\mathds{R}^{d}}
\newcommand{\R}{\mathds{R}}
\newcommand{\dd}{\mathrm{d}}

\newcommand{\vect}[1]{\mathbf{#1}}

\usepackage[thmmarks,  thref, amsmath]{ntheorem}

\theoremstyle{plain}
\newtheorem{defn}{Definition}
\newtheorem{lem}{Lemma}
\newtheorem{theorem}{Theorem}

\newtheorem{prop}{Proposition}
\newtheorem{cor}{Corollary}

\theorembodyfont{\upshape}

\theoremheaderfont{\itshape}
\theorembodyfont{\upshape}

\theoremstyle{break}
\theoremheaderfont{\scshape\bfseries}
\theorembodyfont{\upshape}
\theoremsymbol{}

\theoremstyle{nonumberbreak}
\theoremheaderfont{\scshape\bfseries}
\theorembodyfont{\upshape}
\theoremsymbol{\rule{0.7em}{0.7em}}%
\theorempostwork{\setcounter{case}{0}\setcounter{ppart}{0}}

\newtheorem{proof}{Proof}

\theoremstyle{nonumberplain}
\theoremheaderfont{\itshape\bfseries}
\theorembodyfont{\upshape}
\theoremsymbol{}

\theoremstyle{empty}
\theoremsymbol{}

\usepackage{authblk}

\usepackage{hyperref}
\hypersetup{
    colorlinks = true,
    citecolor=blue,
    filecolor=blue,
    linkcolor=blue,
    urlcolor=blue
}

\title{Doppelg\"anger Model: Emergence of polarization in opinion dynamics}

\author[1]{Vince Campo}
\author[1]{Sebastien Motsch}
\author[2]{Dylan Weber}
\affil[1]{Arizona State University - School of Mathematical and Statistical Sciences, 900 S Palm Walk, Tempe, AZ 85287, USA}
\affil[2]{Changing Character of War Centre, Pembroke College, University of Oxford; Artis International}

\begin{document}

\maketitle
\newpage
\tableofcontents
\begin{abstract}
Over the past decade, contrary to the early popular expectation that large-scale discourse in online communities would foster greater consensus, the large-scale structure of online discourse has been measured to be strongly polarized.  Though it was posited that this effect was driven mainly by the algorithmic curation of content by social platforms (the ``Filter Bubble") it appears that polarization is mainly driven by the tendency of interactions among users to be \textit{homophilic}; users strongly prefer to endorse content and other users that are similar to them.  Users organize into homophilic clusters or ``echo chambers" where interactions within clusters are more likely to be positive and interactions between clusters are more likely to be negative.  This motivates studying polarization in online discourse through modeling as an example of self organized dynamics.  Models of opinion formation on networks have tended to only encode positive interactions (attraction) between users - as a consequence consensus or non-interacting clusters of users are ubiquitous behaviors.  In this research a model of network-based opinion formation is introduced that includes a novel dynamic of negative interaction; instead of modeling negative interaction as repulsion in the opinion space it is modeled as an attraction to a ``Doppelgg\"anger" user.  It is demonstrated that this model exhibits rich behavior. Four distinct regimes of local interaction correspond to different population-wide dynamics - Neutral Consensus, Extreme Consensus, Polarization and Undecided.  In the Polarization regime, local interactions are modeled to be homophilic and the structure of discourse that emerges replicates what is seen in empirical data; clustered users where interactions within clusters is positive and interaction between clusters is negative.
\end{abstract}

\newpage
\section{Introduction}
The internet was and continues to be lauded by many for its role in improving discourse and democratizing information by allowing for more direct discussion among the populace and greatly enhancing access to information generally \cite{rheingold2000virtual, anderson2001universal, fishkin2000virtual, price2002online}.  However at its advent and throughout its history some have cautioned that due to well established psychological tendencies (like confirmation bias and selective exposure), discourse on the internet would actually result in a narrowing of the information landscape of its users.  This view predicted, given the freedom to select information from the virtually infinite array that the internet provides, that users would spend their time on the internet interacting with those who they already agreed with and consuming information from sources that confirmed their prior views.  In other words this view predicted that the structure of online discourse would tend towards "echo chambers" or \textit{homophily} \cite{dean2003net, dahlberg1998cyberspace, buchstein1997bytes, pariser2011filter, kelly2005debate, sunstein2004democracy, republic_com_sunstein}.  The 2016 United States presidential election prompted many empirical studies leveraging large social media datasets looking into the extent to which homophily presents on social media and what effects it might have on its users.  Despite the large variety of methods across this body of work, it is largely unified in its findings; social media has a strongly homophilic structure (see, for example, Figure \ref{fig:empirical}) and this structure contributes to multiple undesirable effects including the narrowing of user information landscapes, enhanced spread of misinformation in homophilic clusters and increased ideological polarization \cite{barbera_birds_2015, barbera_tweeting_2015, garimella_effect_2017, garimella_political_2018, cinelli_echo_2021, monsted_characterizing_2022, schmidt_polarization_2018, bakshy_exposure_2015, bond_quantifying_2015, an_sharing_2014, quattrociocchi_echo_2016, zollo_debunking_2017, bessi_homophily_2016, schmidt_anatomy_2017, gromping_echo_2014, cota_quantifying_2019, cinelli_selective_2020, allcott_social_2017, del_vicario_spreading_2016, vosoughi_spread_2018}.  Many point to the algorithmic curation of content by social platforms as a main culprit, often referred to as the ``Filter Bubble".  This certainly has an effect however, it seems that our own tendency to self-select ideologically aligend content is a more important driver \cite{}.  This motivates the study of the dynamics of online discourse through modeling as an example of \textit{self organized dynamics}; large scale behaviors emerge without a central authority much like the behavior of a flock of birds or shoal of fish \cite{camazine2003self,lopez2012behavioural, castellano_statistical_2009,ballerini2008interaction,cucker_emergent_2007,vicsek2012collective}.

In this arena, models are usually posed in an agent based framework where the "opinion" of an agent is modeled as a continuous value and is influenced, according to a model of local interactions, by those to whom it is connected in a directed network.  Of particular interest is how the topology of the underlying network affects the observed distribution of opinions and vice versa \cite{motsch_heterophilious_2014, saber_consensus_2003, weber_deterministic_2019, li_bounded_2020}.  Often, the interaction models carry an assumption of \textit{local consensus}; if agents interact only with each other than they should agree in some sense.  This assumption can also be interpreted as asserting that there are only attractive forces (equivalently, only positive weights in the interaction network) present among the opinions of the agents.  Additionally, the underlying network topology is often assumed to be homophilic; agents only interact if their opinions are within some threshold of each other.  These assumptions cause the formation of static clusters of agents or a consensus in the long time limit to be a generic behavior \cite{hegselmann_opinion_2002, olfati-saber_consensus_2007, xia_opinion_2011, lorenz_continuous_2007, castellano_statistical_2009, spanos_dynamic_2005, blondel_krauses_2009, ben-naim_unity_2003}. While the interaction networks generated by these types of assumptions have been seen to replicate some of the features of interaction patterns seen in real social media data \cite{schmidt_anatomy_2017, baumann_modeling_2020, del_vicario_spreading_2016}, they usually do not allow for the existence of a repulsive interaction (equivalently, negative weights in the interaction network) between agents.  This is problematic from the modeling point of view as it is well established that a negative interaction can result in a backlash effect where both individuals become further polarized from each other \cite{lord1979biased, zollo_debunking_2017, kuhn1996effects, miller1993attitude, nyhan2010corrections, nyhan2014effective}.  Additionally, under local consensus assumptions, homophily in the interaction network presents as disconnected network components who do not interact, whereas there is evidence that in social media there is significant negative interaction between homophilic clusters \cite{wu_cross-partisan_2021, vasquez_i_2021, gaisbauer_ideological_2021}.

\begin{figure}[H]
  \centering
  \includegraphics[width=0.6\textwidth]{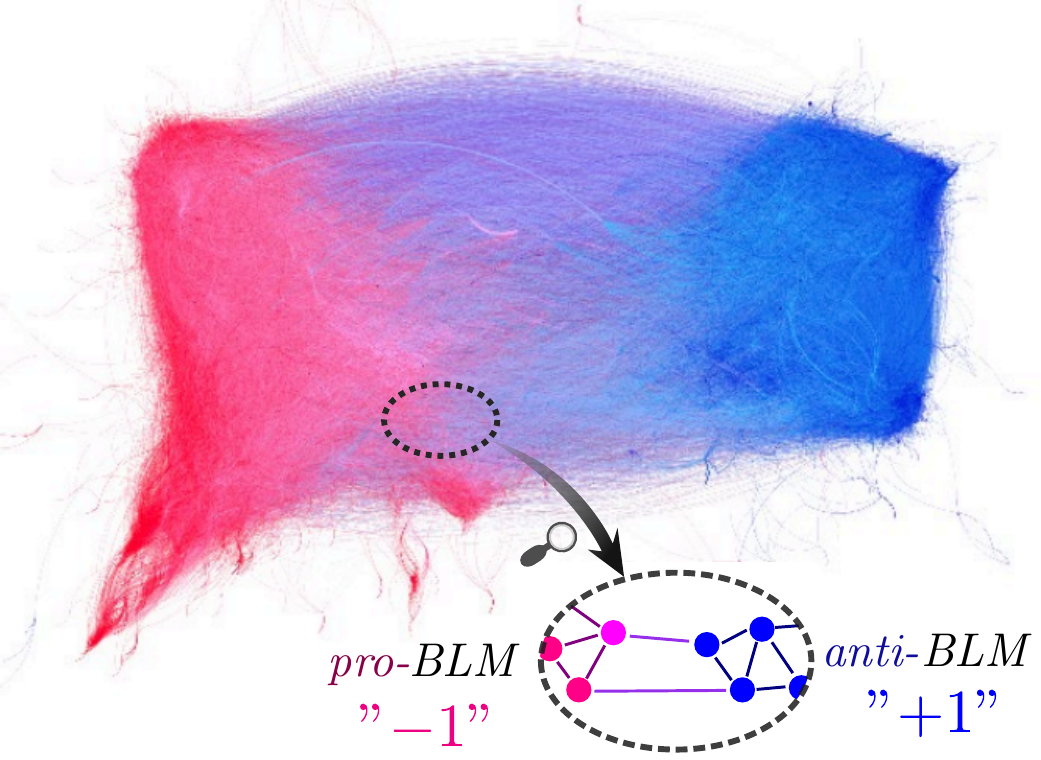}
  \caption{Visualization of the Twitter conversation surrounding the Black Lives Matter movement from May 2023. Nodes in the network are users, edges represent aggregated interactions between users over the data collection period.  Nodes are colored according to their support for or opposition to the Black Lives Matter movement - red indicates support, blue indicates opposition.  Edges are directed and are colored according to the color of the source node.  The network illustrates a strongly homophilic structure.}
  \label{fig:empirical}
\end{figure}

In this paper we introduce and study a class of models that include a mechanism for negative interaction.  Mathematically, inclusion of negative interactions can be problematic; the obvious choice to model them is a repulsion between the interacting agents in the opinion space which can cause the model to diverge.  In this work we introduce a dynamic to address this difficulty; when one agent (the ``interactor") directs a negative interaction at another agent (the ``interactee"), instead of the interactee feeling a ``push" on their opinion away from the interactor's, they feel a ``pull" towards the opinion that is exactly opposite of the interactor's opinion.  Since the interactee is being attracted to an opinion that is similar but opposite of the opinion of the interactor, we dub this class of models the "Doppelg\"anger" model.
We first study a version of the model where the interaction network is static.  We find that in all cases the model remains bounded and further show that if the adjacency matrix associated to the interaction network has a dominant eigenvalue then the model converges to the eigenvector associated to the dominant eigenvector (known as the Perron-Frobenius eigenvector).

We then generalize the model to the case where the interaction network is dependent on the distribution of opinions.  Specifically, we introduce the notion of an \textit{interaction function} that encodes whether an interaction between two agents is positive or negative depending on the distance between their opinions.  We again find that in all cases the model remains bounded.  Interestingly, depending on the shape of the interaction function, we find that the model displays a range of possible behaviors including convergence to a polarized state, convergence to neutral consensus, convergence to extreme consensus, and a failure to converge that could be interpreted as ongoing deliberation.  We prove rigorously that this nonlinear model converges to a consensus in the case that the interaction function is strictly positive.  To study how different interaction functions result in different regimes we introduce a prototypical interaction function in the form of a parametrized Morse potential and examine the long time behavior of the model in the parameter space of the interaction function.  In particular, in the case that the interaction function models homophilic interactions (interactions are positive at short opinion ranges and negative at long ranges), we find that the distribution of opinions converges to a polarized state and that the interaction network converges to a structure that is similar to what is observed in empirical data; homophilic clusters where interaction among clusters is positive and interaction between clusters is negative.

\section{Opinion formations with attraction and repulsion}
Our main aim is to introduce a model of opinion formation that facilitates a better understanding of the polarization of opinions as well as consensus. As described, most analysis of opinion formation models have been conducted on models where only attractive forces are present among the agents.  We aim at introducing repulsion in models of opinion formation; specifically we aim to include the possibility of edges in the interaction network with negative weights.  Denoting by $s_i(t)$ the opinion of an individual $i$ at time $t$ ($s_i$ could be a scalar or a vector), we consider an adjacency matrix $A$ with weights $a_{ij}$ that could be positive and negative, modeling the attraction of an opinion if $a_{ij}>0$ or repulsion if $a_{ij}<0$. Rather than adding an attraction/repulsion that is proportional to $s_j-s_i$, we propose dynamics where the opinion $s_i$ is actually attracted to the opposite opinion of $s_j$, i.e. $-s_j$, we dub the class of models with this interaction mechanism, the \textit{Doppelg\"anger model}.

\begin{figure}[H]
    \centering
    \includegraphics[width=0.6\linewidth]{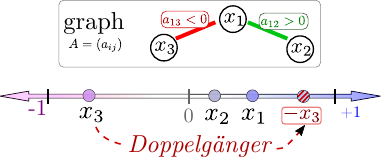}
  \caption{A schematic of the Doppelg\"anger interaction.  The agent $x_{1}$ is connected by a positive weight to $x_{2}$ and a negative weight to $x_{3}$ and therefore is attracted to the opinion of $x_{2}$ and the \textit{opposite} opinion of $x_{3}$.}
  \label{fig:Doppelganger_interaction}
\end{figure}

\subsection{Fixed-graph Doppelg\"anger Model}
\begin{defn}[fixed-graph Doppelg\"anger]
  Given a collection of $N$ agents, let $s_{i}\in \R^d$ represent the opinion of the i$^{th}$ agent and adjacenty matrix $A$. The \textbf{fixed-graph Doppelg\"anger model} is defined by the dynamics
\begin{equation}
  \label{eq:fixed_graph_Doppelganger}
  s_i'= \sum_{j=1}^N a_{ij} s_j   - ||\textbf{s}|| s_i \quad\text{with} \quad ||\textbf{s}||  := \sqrt{\frac{1}{N}\sum_{i=1}^N |s_i|^2}
\end{equation}
where $|s_i|$ is the normalized Euclidean norm.
% \begin{equation}
%   |s_i|^2 = \frac{1}{d}\sum_{k=1}^d |s_i^{(k)}|^2.
%   \label{eq:normalized_eucliean_norm}
% \end{equation}
\end{defn}

When $a_{ij}>0$, the opinion of the $i$th agent, $s_i$, is attracted to the opinion of the $j$th agent, $s_j$, whereas when $a_{ij}<0$ then $s_i$ is attracted to $-s_j$. We dub the dynamics the Doppelg\"anger model as each opinion $s_i$ has a mirror opinion $-s_i$ that has an opposite effect.  The fixed-graph Doppelg\"anger model does not necessarily converge and moreover consensus is not always a stable solution.

\begin{figure}[H]
  \centering
  \includegraphics[width=\linewidth]{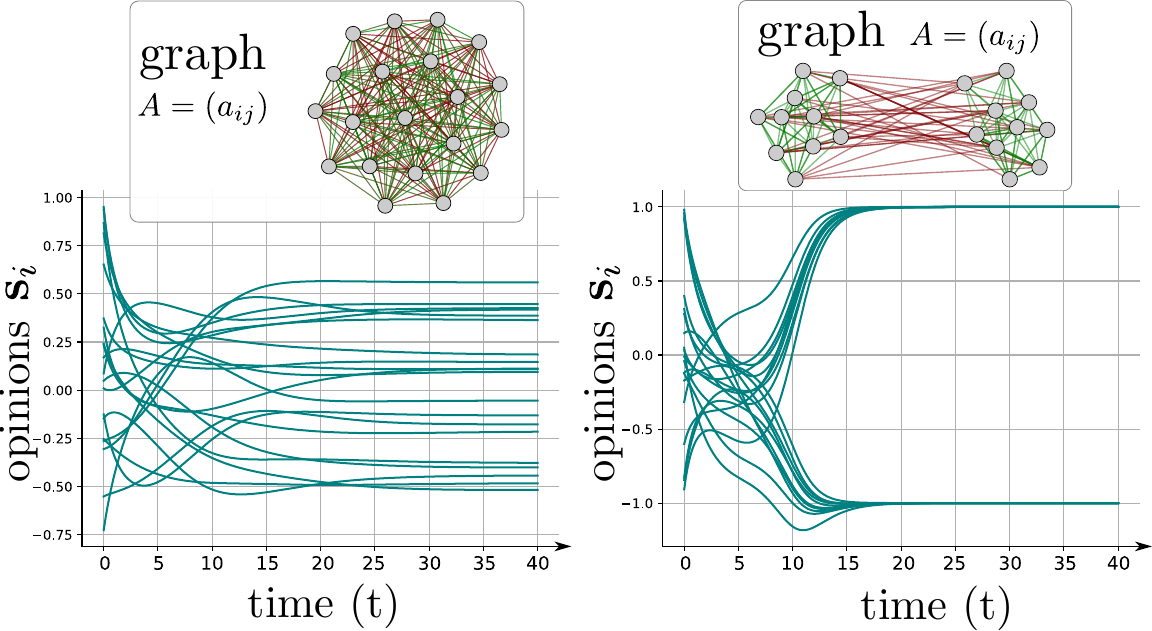}
  \caption{Evolution of the opinions $s_i(t)$ for the Doppelg\"anger dynamics \eqref{eq:fixed_graph_Doppelganger} using two different adjacency matrices $A=(a_{ij})$. In the left plot the interaction network models a ``normative" social interaction network and the opinion distribution converges to a configuration that is spread out within a moderate opinion range.  In the right plot the interaction network models a homophilic interaction network similar to what is seen in empirical social media data.  In this case, the opinions converge to two polarized clusters (corresponding to the homophilic network components) at the extreme ends of the opinion spectrum.  In both cases, the dynamics converge to the dominant eigenvector of the adjacency matrix, $A$. }
  \label{fig:Doppelganger}
\end{figure}

\subsection{Boundedness of the Solution}
Before investigating the long-time behavior dynamics of the Doppelg\"anger model, we first show that the solution remains bounded.
\begin{prop}
  \label{prop:dopp_bounded}
  The solutions to the fixed-graph Doppelg\"anger \eqref{eq:fixed_graph_Doppelganger}  remain uniformly bounded in time.
\end{prop}
\begin{proof} Consider a constant $C$ such that $|a_{ij}|\leq C$. For $\textbf{u}=(u_i)_{i=1:N}$ and $\textbf{v}=(v_i)_{i=1:N}$, let's define the scalar product:
  \begin{displaymath}
    \langle \textbf{u},\textbf{v}\rangle  = \frac{1}{N}\sum_{i=1}^N \langle u_i,v_i\rangle _{\mathbb{R}^d}
  \end{displaymath}
  where $\langle .,. \rangle_{\mathbb{R}^d}$ is the standard scalar product on $\mathbb{R}^d$. In particular, $\sqrt{\langle s,s\rangle }=\|s\|^2$. We deduce:
  \begin{eqnarray}
    \label{eq:dop_bounded}
    \frac{1}{2}\frac{d}{dt} \|\textbf{s}\|^2 &=& \frac{1}{2N}\frac{d}{dt} \sum_{i=1}^N |s_i|^2  \\&=& \frac{1}{N}\sum_{i=1}^N \langle s_i, \sum_{j=1}^Na_{ij}s_j - ||\textbf{s}||s_i\rangle _{\mathbb{R}^d} \\
                                         &=& \frac{1}{N}\sum_{i,j=1}^N a_{ij}\langle s_i, s_j\rangle _{\mathbb{R}^d}  - \frac{||\textbf{s}||}{N}\sum_{i=1}^N |s_i|^2 \\
                                         %&=& \langle As,s\rangle  - \|s\|^3 \\
                                         &\leq& C\|\textbf{s}\|^2 \;-\; \|\textbf{s}\|^3,
                                         %&\leq& \|s\|^2(C_\phi-\|s\|).
  \end{eqnarray}
  using the Cauchy-Schwarz inequality. Thus, the norm $\|\textbf{s}\|$ is decaying when $\|\textbf{s}\|>C$. Therefore, $\|\textbf{s}(t)\|\leq \max(\|\textbf{s}(t=0)\|,C)$ for all $t>0$.
\end{proof}

The result easily generalizes to the non-linear Doppelg\"anger dynamics \eqref{eq:Doppelganger} assuming that the interaction function $\phi$ is bounded. Indeed, the (non-constant) matrix $A(t)= (a_{ij}(t))_{ij}$ where $a_{ij}$ given by \eqref{eq:Doppelganger} is uniformly bounded in time when $\phi$ is bounded.

\subsection{Convergence to equilibrium}
We now turn our attention to the convergence of the fixed-graph Doppelg\"anger model \eqref{eq:fixed_graph_Doppelganger}.  The dynamics converge to an equilibrium under the assumption that $A$ has a dominant eigenvalue. % (see also figure \ref{fig:dominant_eigen}).

\begin{defn}
  A square matrix $A$ has a dominant eigenvalue if there exists an eigenvalue $r$ of $A$ satisfying:
  \begin{equation}
    \label{eq:perron_eigen}
    |\lambda|<r \quad \text{ for any other eigenvalue } \lambda { of } A.
  \end{equation}
  The eigenvalue $r$ is referred to as the Perron-Frobenius eigenvalue. An eigenvector associated with $r$ is called a Perron-Frobenius eigenvector and denoted $\textbf{u}_r$.
\end{defn}

\begin{theorem}
  \label{thm:Doppelganger_fixed_graph}
  Consider the fixed-graph Doppelg\"anger model \eqref{eq:fixed_graph_Doppelganger} in dimension $d=1$. Assume that $A$ has a dominant eigenvalue $r$ and the initial condition $s(t=0)$ has a non-zero component along the associated Perron-Frobenius eigenvector. Then the solution $s(t)$ converges (exponentially fast) to a Perron-Frobenius eigenvector $\textbf{u}_r$ of $A$.
\end{theorem}
\begin{proof}
  First, we decompose the dynamics into two parts using the eigenspaces of $A$. We denote $E_r$ the eigenspace related to the dominant eigenvalue (i.e. $E_r = \text{Span}(\textbf{u}_r)$) and  $E_s$ the union of all the other eigenspaces (i.e. $E_s= E_{\lambda_1}\bigoplus E_{\lambda_2}...$). Thus,
  \begin{equation}
    \label{eq:decomposition_Rn}
    \R^n = E_{r}\oplus E_{s}.
  \end{equation}
The solution to the Doppelg\"anger model has a unique decomposition:
\begin{equation}
  \label{eq:decompositin_s}
  \textbf{s}(t) = \alpha(t) \textbf{u}_r + \textbf{v}(t) \, \qquad \text{with}\quad \alpha(t) \textbf{u}_r \in E_r,\, \textbf{v}(t) \in E_s.
\end{equation}
Plugging this decomposition into the dynamics \eqref{eq:fixed_graph_Doppelganger}, we find:
\begin{equation*}
  \alpha' \textbf{u}_r + \textbf{v}' = r \,\alpha \textbf{u}_r + A\textbf{v} \;-\; ||\textbf{s}||\big(\alpha \textbf{u}_r + \textbf{v}\big) \;\; \\
\end{equation*}
\begin{equation}
  \label{eq:dec}
  \Rightarrow \;\; \left\{
    \begin{array}{rcccc}
      \alpha' &=& r \alpha & - &  ||\textbf{s}||\alpha \\
      \textbf{v}' &=& A_{|E_s} \textbf{v} &-& ||\textbf{s}|| \textbf{v}
    \end{array}
  \right.
\end{equation}
The two equations are linked by the non-linear term $||\textbf{s}||$. But it is a 'gentle' non-linearity, if one could control this damping term, one could analyze the long term behavior of the solution in each eigenspace respectively.

In a second step, we try to find some lower bound on $||\textbf{s}||$. We already know that $||\textbf{s}|| \leq C$ thank to proposition \ref{prop:dopp_bounded}. Using \eqref{eq:dec}, we can further write the solution as (Duhamel's formulation):
\begin{equation}
  \label{eq:duhamel}
  \alpha(t) = \exp\left(r t - \int_0^t ||\textbf{s}(z)||\,d z\right)\alpha(0).
\end{equation}
Without loss of generality, we can assume that $\alpha(t)>0$. Since ${\bf S}(t)$ remains bounded (proposition \ref{prop:dopp_bounded}), we also have that $\alpha(t)\leq C$ and therefore:
\begin{equation}
  \label{eq:inequality_int_mu}
  r t - \int_0^t ||\textbf{s}(z)||\,d z \leq C \quad \Rightarrow \quad r t - C \leq \int_0^t ||\textbf{s}(z)||\,d z.
\end{equation}

Knowing this bound, we now turn our attention to the other components $\textbf{v}(t)$ in \eqref{eq:dec}. Using the spectral gap satisfies by $A$, there exists $\varepsilon>0$ and a norm $\|.\|_{\sim}$ such that:
\begin{equation}
  \label{eq:spectral_gap_A_E}
  \|A_{|E_s}\|_{\sim} \leq r-\varepsilon.
\end{equation}
Therefore, we deduce:
\begin{align}
  \|\textbf{v}'\| &\leq (r-\varepsilon)\|\textbf{v}\| - ||\textbf{s}(t)||\cdot \|\textbf{v}\| \\  \|\textbf{v}(t)\| &\leq \exp\left((r-\varepsilon)t - \int_0^t ||\textbf{s}(z)||\,d z\right) \nonumber \\
                                                                 \|\textbf{v}(t)\| &\leq \exp\left(C-\varepsilon t\right) \stackrel{t \rightarrow +\infty}{\longrightarrow} 0, \label{eq:cv_v}
\end{align}
where we use the inequality \eqref{eq:inequality_int_mu}.

To finish the proof, we need to prove the convergence in time of the component $\alpha(t)$. Thanks to the previous result \eqref{eq:cv_v}, we can write:
\begin{equation}
  \label{eq:mu_dec}
  \mu(\textbf{s}(t)) = \|\alpha(t) \textbf{u}_r\| + \delta(t)  \quad \text{with} \quad \delta(t) \stackrel{t \rightarrow +\infty}{\longrightarrow} 0.
\end{equation}
Indeed,
\begin{align*}
  |\delta(t)| &= \left|\|\alpha(t)\textbf{u}_r + \textbf{v}(t)\| - \|\alpha(t)\textbf{u}_r\|\right| \\
  |\delta(t)| &\leq \|\textbf{v}(t)\| \stackrel{t \rightarrow +\infty}{\longrightarrow} 0.
\end{align*}
Furthermore, since the decay of $\|\textbf{v}(t)\|$ toward zero is exponential \eqref{eq:cv_v}, the perturbation $\delta(t)$ also satisfies:
\begin{equation}
  \label{eq:decay_delta}
  \int_0^{+\infty} |\delta(s)| \,d s \leq C.
\end{equation}
Without loss of generality, we have assume that $\alpha(0)\neq 0$ and therefore $\alpha(t)>0$. We can also assume that the eigenvector $\textbf{u}_r$ is normalized, i.e. $\|\textbf{u}_r\|=1$. Thus, we obtain:
\begin{equation}
  \label{eq:behavior_alpha}
 ||\textbf{s}(t)|| = \alpha(t) + \delta(t) \qquad \text{with } \delta(t) \stackrel{t \rightarrow +\infty}{\longrightarrow} 0.
\end{equation}
We can now go back to the equation satisfied by $\alpha(t)$ \eqref{eq:dec} to deduce:
\begin{equation}
  \label{eq:alpha_ode_perturb}
  \alpha' = (r-\delta(t)) \alpha - \alpha^2 \quad \text{with } \delta(t) \stackrel{t \rightarrow +\infty}{\longrightarrow} 0.
\end{equation}
Thus, $\alpha(t)$ satisties a Bernouilli differential equation for which there exists an explicit solution.  We conclude using the result from lemma \eqref{lem:Bernouilli_ODE} included in the appendix.
\end{proof}

\section{Doppelg\"anger dynamics with evolving graphs}
The Doppelg\"anger model is naturally extended to a ``fully" non-linear version by making the weights $a_{i,j}$ of the adjacency matrix $A$ dependent on the distribution of opinions. A natural choice to allow for the modeling of ``local" biases in the agents is to impose that the weight in the interaction network between agents $i$ and $j$, $a_{ij}$, is a function of the difference of opinion between agents $i$ and $j$.  For example, if the interaction function is positive in a neighborhood of $0$ and negative outside of that neighborhood this could model the the local biases of selective exposure/confirmation bias and backlash effect.
\subsection{Fully non-linear Doppelg\"anger Model}
\begin{defn}[The Doppelg\"anger model]
Given a collection of $N$ agents, let $s_{i}\in \R^d$ represent the opinion of the i$^{th}$ agent. Let $\phi: \R \rightarrow \R$ be the interaction function. The \textbf{Doppelg\"anger model} is defined by the dynamics
\begin{equation}
\label{eq:Doppelganger}
s_i'= \sum_{j=1}^N a_{ij} s_j   - ||\textbf{s}|| s_i.
\end{equation}
where
\begin{equation}
   a_{ij} = \frac{\phi(|s_{j} - s_{i}|)}{\sum_{k=1}^N|\phi(|s_{k} - s_{i}|)|} \quad \text{and} \quad ||\textbf{s}|| = \sqrt{\frac{1}{N}\sum_{i=1}^N |s_i|^2}.
\end{equation}
\end{defn}
Notice that the interaction function $\phi$ is not always positive. As long as the function $\phi$ is bounded, the opinions $s_i(t)$ remain bounded using the inequality \eqref{eq:dop_bounded}. The convergence of the dynamics toward an equilibrium is more delicate, the dynamics could very well oscillate. The investigation of the dynamics will mainly focus on the interplay between the support of the function $\phi$, the spectrum of the corresponding adjacency matrix $A$ and the convex hull of the configuration of opinions.  For example, in the case that the interaction function is positive in a neighborhood of $0$ and negative outside of that neighborhood, the model displays a "consensus trap" - if the convex hull of the initial configuration is sufficiently small (such that the difference between all agent opinions lies in the set of the domain of $\phi$ where $\phi$ is positive) then consensus is the only outcome possible.

\begin{cor}
  \label{cor:dop_bounded}
  The solution to the non-linear Doppelg\"anger \eqref{eq:Doppelganger} with $\phi$ bounded remains uniformly bounded in time.
\end{cor}
\begin{proof}
  It is sufficient to notice that the coefficient $a_{ij}$ are uniformly bounded with:
  \begin{displaymath}
    |a_{ij}|\leq \max(1,M) \quad  \text{where } \; M=\sup_{r \geq 0} |\phi(r)|.
  \end{displaymath}
  Thus, we can apply the same computation as in proposition \ref{prop:dopp_bounded} and conclude.
\end{proof}

\begin{theorem}
  \label{thm:cv_phi_positive}
  Suppose that the attraction/repulsion function $\phi$ is strictly positive. Then the (fully non-linear) Doppelg\"anger model \eqref{eq:Doppelganger} converge to a consensus, i.e. there exists $s_*$ such that:
  \begin{displaymath}
    s_i(t) \stackrel{t \rightarrow +\infty}{\rightarrow} s_*.
  \end{displaymath}
\end{theorem}
\begin{proof}
  Since $\phi>0$, the coefficients $(a_{ij})_j$ form a convex combination for any $i$. Thus, we can rewrite the dynamics as:
  \begin{equation}
    \label{eq:s_trick}
    s_i' = \sum_{j=1}^N a_{ij} (s_j-s_i)   \;+\; (1- ||\textbf{s}||) s_i
  \end{equation}
  Using an integrating factor, we introduce the new coordinates:
  \begin{equation}
    \label{eq:def_z}
    z_i(t) = e^{\nu(t)}s_i(t) \qquad \text{with } \nu(t)= \int_{0}^t -1+||\textbf{s}(u)||\,du
  \end{equation}
  Taking the time derivative of $z_i(t)$, the new coordinates satisfy:
  \begin{equation}
    \label{eq:z_eq}
    z_i'= \sum_{j=1}^N a_{ij} (z_j-z_i).
  \end{equation}
  The coefficients $a_{ij}$ are time-dependent but since $\{s_i(t)\}_i$ remains uniformly bounded, i.e. there exists $L>0$ such that $s_i(t)\in[-L,L]$ thanks to (corollary \ref{cor:dop_bounded}) we can find a lower-bounded:
  \begin{displaymath}
    a_{ij}(t) = \frac{\phi(|s_i(t)-s_j(t)|)}{\sum_{k=1}^N \phi(|s_i(t)-s_k(t)|)} \geq \frac{\underline{m}}{N\cdot\overline{m}}>0,
  \end{displaymath}
  where $\underline{m}$ and $\overline{m}$ are the lower and upper bound (respectively) of $\phi$ over $[0,2L]$. Using the result from \cite{motsch_heterophilious_2014,hegselmann2002opinion} we conclude that $\{z_i(t)\}_i$ converges to a consensus, i.e. there exists $z_*$ such that:
  \begin{equation}
    \label{eq:cv_z}
    z_i(t) \stackrel{t \rightarrow +\infty}{\rightarrow} z_*.
  \end{equation}
  As a consequence the variance of $\{s_i(t)\}_i$ converges to zero:
  \begin{eqnarray*}
    \text{Var}[\{s_i(t)\}_i] &=& \frac{1}{N^2} \sum_{1\leq i,j\leq N}|s_j(t)-s_i(t)|^2 \\&=& \frac{1}{N^2} e^{v(t)}\sum_{1\leq i,j\leq N} |z_j(t)-z_i(t)|^2 \stackrel{t \rightarrow +\infty}{\rightarrow} 0,
  \end{eqnarray*}
  using that $e^{v(t)}$ remains uniform bounded and the convergence of $\{z_i(t)\}_i$.

  To conclude the proof, we only need to show that mean of $\{s_i(t)\}_i$ converges. Denote $\overline{s}(t)=\frac{1}{N}\sum_{i=1}^N s_i(t)$, we find:
  \begin{eqnarray*}
    \overline{s}'  &=& \frac{1}{N}\sum_{i=1}^N\left(\sum_{j=1}^N a_{ij} (s_j-s_i)   \;+\; (1- ||\textbf{s}||) s_i\right) \\&=& \sum_{1\leq i,j\leq N} \frac{a_{ij}}{N} (s_j-s_i)   \;+\; (1- ||\textbf{s}||) \overline{s}.
  \end{eqnarray*}
  Since the coefficients $a_{ij}$ are upper-bounded, we notice notice that:
  \begin{equation}
    \label{eq:small_term}
    \epsilon_1(t):=\sum_{1\leq i,j\leq N} \frac{a_{ij}(t)}{N} (s_j(t)-s_i(t))\stackrel{t \rightarrow +\infty}{\rightarrow} 0.
  \end{equation}
  Moreover:
  \begin{eqnarray*}
    ||\textbf{s}(t)|| &=& \sqrt{\frac{1}{N}\sum_{i=1}^N|s_i(t)|^2} \\&=& \sqrt{\frac{1}{N}\sum_{i=1}^N (|\overline{s}(t)|^2 + |s_i(t)-\overline{s}(t)|^2)}\\
     &=& \sqrt{|\overline{s}(t)|^2 + \text{Var}[\{s_i(t)\}_i]} = |\overline{s}(t)| + \epsilon_2(t),
  \end{eqnarray*}
  with $\epsilon_2(t) \stackrel{t \rightarrow +\infty}{\rightarrow} 0$.
  Thus, the ODE satisfies by $\overline{s}(t)$ can be written in the form:
  \begin{equation}
    \label{eq:ODE_mean_s}
    \overline{s}' = \epsilon_1 + (1-\epsilon_2)\cdot\overline{s} - |\overline{s}|\overline{s}.
  \end{equation}
  Thus, we obtain a small perturbation of the Bernouilli differential equation as in Lemma \ref{lem:Bernouilli_ODE}. Thus, we conclude that $\overline{s}(t)$ convergences, i.e. there exists $s_*$ such that: $\overline{s}(t) \stackrel{t \rightarrow +\infty}{\rightarrow} s_*$. Since the variance of $\{s_i(t)\}_i$ is converging to zero, we deduce that $s_i(t) \stackrel{t \rightarrow +\infty}{\rightarrow} s_*$ for all $i$.
\end{proof}

\subsection{Morse Potential}
Before trying to analyze the non-linear Doppelg\"anger dynamics, we numerically explore the different possible regimes of behavior depending on the interaction function $\phi$.  In order to smoothly vary the shape of the interaction function with the goal of identifying possible phase transitions, we introduce a parameterized Morse potential as a prototypical interaction function:
\begin{equation}
  \label{eq:morse}
  \phi(r) = a e^{-r} + b e^{-r^2}.
\end{equation}

Depending on the two parameters $a$ and $b$, the function $\phi$ has various shapes.  To explore different possible regimes of the model behavior we examine how the model behavior varies in the parameter space of the interaction function.  Notice that due to the normalization in the definition of $a_{ij}$, changing  $(a,b)$ to $(\alpha a,\alpha b)$ with $\alpha>0$ has no influence on the dynamics. For this reason, we only illustrate the More potential using $(a,b)$ on the unit circle. Where we see the following results:
\begin{itemize}
\item $\theta = \frac{\pi}{4}$ Figure \ref{fig:2D_Phase_Diagrams}(A): the interaction function satisfies $\phi>0$ (attraction). The dynamics convergence to an \textit{extreme consensus} $\pm 1$.
\item $\theta = \frac{3\pi}{4}$ Figure \ref{fig:2D_Phase_Diagrams}(B): the function $\phi$ is positive at close range and negative at larger range ({\it homophilious dynamics}). The dynamics converge to a polarized state.
\item $\theta = \frac{5\pi}{4}$ Figure \ref{fig:2D_Phase_Diagrams}(C): here $\phi<0$ (repulsion). The dynamics convergence to a {\it neutral} consensus $0$.
\item $\theta = \frac{7\pi}{4}$ Figure \ref{fig:2D_Phase_Diagrams}(D): the function $\phi$ is now negative at close range and positive at larger range ({\it heterophilious dynamics}). The dynamics is no-longer converging to an equilibrium and we observe oscillatory behavior.
\end{itemize}

\begin{figure}[H]
    \centering
    \includegraphics[width =\linewidth]{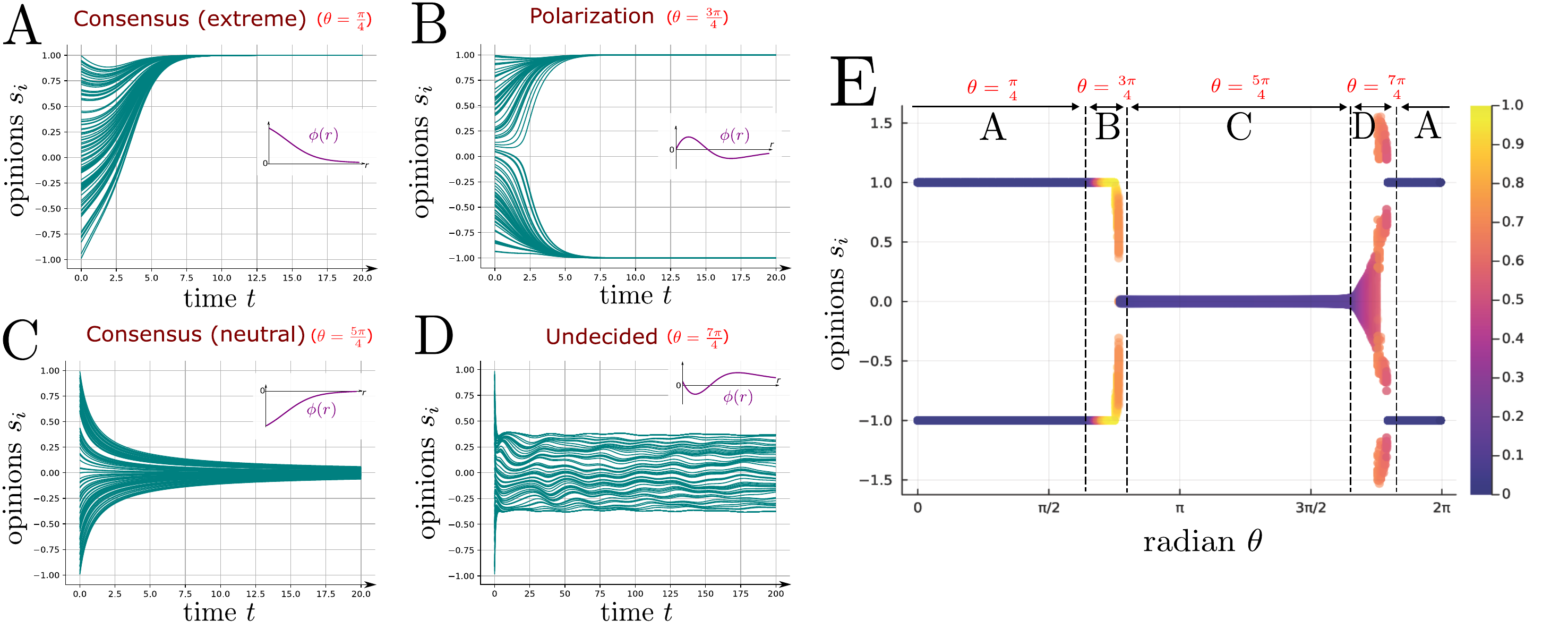}
    \caption{(A)-(D) Illustration of the fully non-linear Doppelg\"anger dynamics \eqref{eq:Doppelganger} using the Morse potential \eqref{eq:morse} $N=100$ agents.\\(E) Figure of the final positions of the Doppelg\"anger dynamics utilizing Morse Potential \eqref{eq:morse} $N=100$ agents and 30 simulations.}
    \label{fig:2D_Phase_Diagrams}
\end{figure}

Utilizing the equation \eqref{eq:morse} we set the parameters $a,b$ by varying $\theta\in [0,2\pi]$, where $a = \cos\theta$ and $b = \sin\theta$. By varying the parameters, we explore how the dynamics transition between the states described in Figure \ref{fig:2D_Phase_Diagrams} (A-D). We propose 30 simulations over 100 agents (N=100) where we plot the final opinions and analyze the spread of these opinions. Notably, the state of Consensus (neutral/extreme) are both more apparent with this model. We can classify these stages by looking at the spread which is relatively small ($\approx$ 0) where in the case of Consensus (extreme), through numerous simulations, can tend to either side depending on the proximity of the overall opinions to one extreme. Within Figure \ref{fig:2D_Phase_Diagrams}(E), to make a distinction from Consensus Extreme to Polarization although differs in variance. More interestingly, the stage of polarization is where we see a split in the final opinions where rather having a consensus we see that opinions attract to both extremes thus causing the spread to increase. From Polarization as it transitions the opinions tend toward a neutral state as the measure of spread rapidly decreases and achieves a consensus.  Lastly, the smaller increase in spread as degree increases correlates to the Undecided state. This creates a model where the opinions are not spread between two extremes rather cluster around a neutral opinion with a small but existing spread in opinion until it ultimately trends toward once more to consensus extreme.

\section{Exploration in higher dimension}
We now explore if the results obtained in the previous sections for one dimensional opinions ($s_{i} \in \R$) hold for $d$ dimensional opinions under the same dynamics.  Our boundedness results, Proposition \ref{prop:dopp_bounded} and Corollary \ref{cor:dop_bounded}, hold independent of dimension. However, our proofs concerning convergence in both the fixed-graph and fully non-linear cases assume one dimensional opinions.

\subsection{Convergence to equilibrium of the fixed-graph Doppelg\"anger dynamics}
As in our one dimensional exploration, we first consider the case where the interaction network among agents is fixed and does not depend on the opinion distribution among the agents.  We find that much of our analytical work in the one dimensional case transfers here as the vector of the $i$th coordinates of all of the agents, $\textbf{s}_{\dot, i} = (s_{1,i}, s_{2,i},...,s_{N,i})$ where $i \in \set{1,...,d}$ evolves according to the one dimensional dynamics shown in Figure \ref{fig:3D_phase_diagrams} (A-D).

\subsection{Dynamics of Phase Transition for higher dimensions}
Here we simulate 50 vectors of opinions $s^{1}_{i}, s^{2}_{i}$ and plot the final positions while measuring the square root of the standard deviation over both opinions separately. The phase diagram is consistent with its two dimensional counterpart in Figure \ref{fig:2D_Phase_Diagrams} (E). Differing from the two dimensional figure, we take the sum of spread of the variables thus our measure takes on zero to two. Interestingly we see that the opinions for Consensus Extreme take on opinions that create a circle for each of the 30 simulations performed. Although, we can differentiate better the stages of Undecided and Polarization based on their spread where Polarization takes values closer to 2 (or 1 in both x and y) while, undecided sees a smaller spread. Notably the dynamics that are shown here utilize the Frobenius norm when defining the term $||\textbf{s}||$.

We can utilize additional matrix norms to convey dynamics that trend toward the extremes of opinions $s_{i}^{k}$ for $k=1,\hdots, N$. By examining Figure \ref{fig:3D_phase_diagrams} (E), we explore the possibility of other behaviors by varying the function given by the parameter $||\textbf{s}||$. The Frobenius norm gives us the shape of a circle where as we keep the agents bounded within $[-1,1]$ then the dynamic yields $||s|| = |s_{i}^{1}|^{2}+|s_{i}^{2}|^{2} = 1$. We can modify this by redefining $||s||$ by using the infinity-norm such that we have $||s|| = \max\{|s_{i}^{1}|,|s_{i}^{2}|\}$. This behavior should yield maximum values in both the $s^{1}, s^{2}$, which would result in a square around the coordinate pairs of $\{1,-1\}$. Lastly, we have the 1-norm which takes in the sum of the agents which we can define as $|s_{i}^{1}| + |s_{i}^{2}|=1$, where the shape takes the form as a square although the points are on the intercepts of both opinions on $\{-1,1\}$. Through variation, this would determine the final opinions of the agents and not the four behaviors as seen in Figure \ref{fig:3D_phase_diagrams} (A-D).

 \begin{figure}[H]
    \centering
    \includegraphics[width=\textwidth]{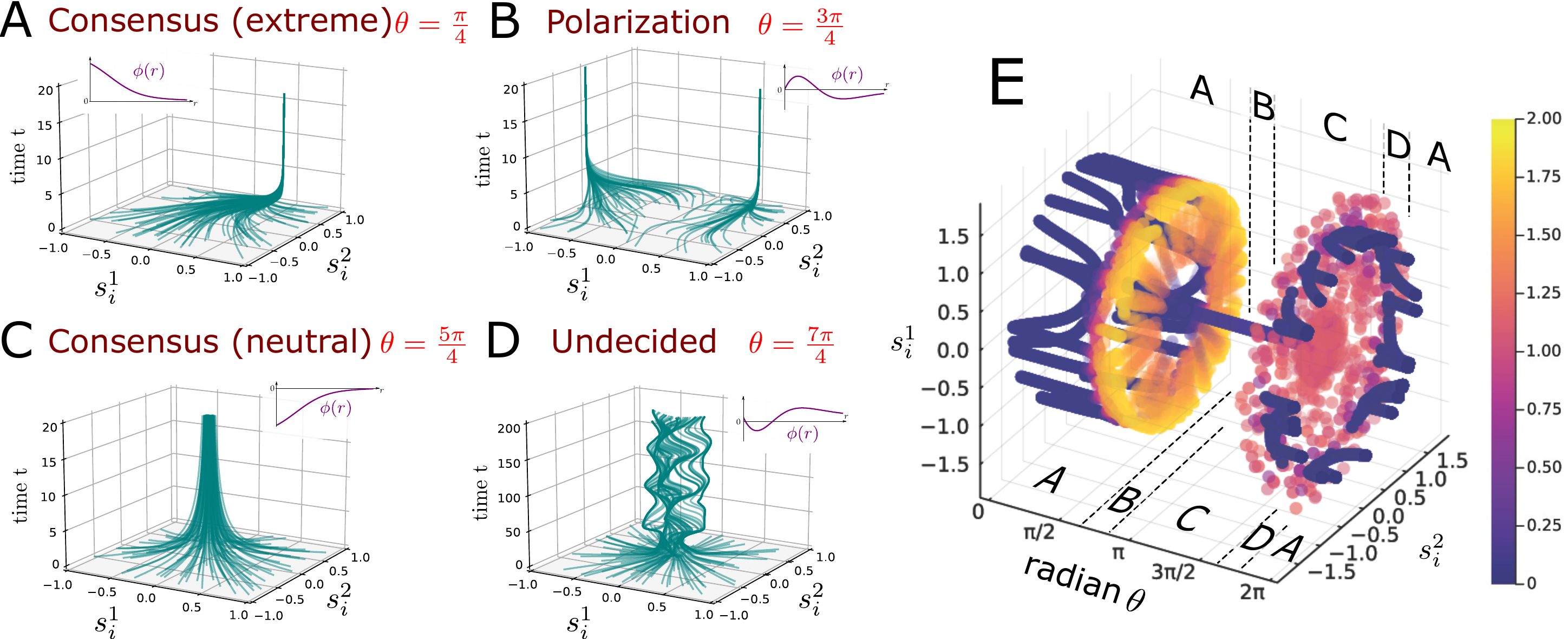}
    \caption{(A)-(D) Illustration of the fully non-linear Doppelg\"anger dynamics \eqref{eq:Doppelganger} in dimension $2$ using the Morse potential equation \eqref{eq:morse} using $N=100$ agents.\\
    (E) Phase Transitions by varying $\theta$}
    \label{fig:3D_phase_diagrams}
\end{figure}

\section{Conclusion}

In this manuscript, we introduce a new type of opinion model where opinions may be repelled not directly but rather by being attracted to the opposite opinions. We show numerically and in some cases analytically the emergence of distinct equilibrium such as polarization or consensus. Of particular interestd is the emergence of a phase transition between this different equilibrium as we gradually vary the interaction function.

Several directions require further exploration. On the empirical side, it is necessary to estimate the most appropriate form of the interaction function, $ϕ$, from data. Analytically, we were only able to prove convergence of the dynamics to an extreme consensus when $ϕ>0$. Therefore, studying other cases remains an open question. In terms of modeling, there are multiple ways to combine different opinions in higher dimensions. For instance, do two users disagree if only one of their opinions differs, or is disagreement necessary across several components?

In another context, it would be valuable to study the limiting case of an infinite number of opinions, $N→+∞$. In this case, the Doppelg\"anger model would reduce to a continuity equation on the density function $ρ$:
\begin{equation}
  ∂_t ρ(s,t) + ∇⋅\Big( (\overline{v}_{ρ}(s)-\overline{e}_{ρ}⋅s)ρ(s,t)\Big) = 0,
  \label{eq:PDE_rho}
\end{equation}
with:
\begin{eqnarray}
\overline{v}_{ρ}(s) &=& \frac{∫_{\tilde{s}∈ℝ^d}ϕ(s-\tilde{s})\tilde{s}ρ(\tilde{s})\,\dd \tilde{s}}{∫_{\tilde{s}∈ℝ^d}|ϕ(s-\tilde{s})|ρ(\tilde{s})\,\dd \tilde{s}} \\
  e_{ρ}^2 &=& ∫_{\tilde{s}∈ℝ^d} |\tilde{s}|^2 \rho(\tilde{s})\,\dd \tilde{s}.
\end{eqnarray}
Proving this limit as $N→+∞$ is quite challenging (see propagation of chaos [REF]). However, studying the long-term behavior of systems from a PDE perspective could be simpler and may provide greater insight into the phase transitions of the dynamics.

\newpage
\section{Appendix}

\subsection{Non-dominant Eigenvalue Fixed Graph}
Here we look at a unique case where we have no dominant eigenvalue. This creates a dynamic where the solution does not converge to a consensus or stable solution.
\begin{figure}[H]
  \centering
  \includegraphics[width=0.6\linewidth]{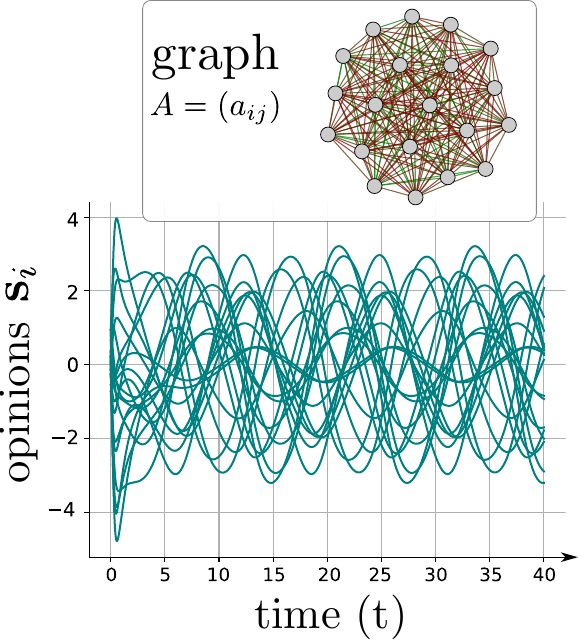}
  \caption{An example of the fixed graph Doppelg\"anger dynamics where the underlying network does not have a dominant eigenvector - the model does not converge (but remains bounded) and demonstrates oscillatory behavior.}
  \label{fig:Doppelganger_noconverge}
\end{figure}

\subsection{Bernoulli Differential Equation Solution}
\begin{lem}
  \label{lem:Bernouilli_ODE}
  Suppose $y(t)$ satisfies the following differential equation:
  \begin{equation}
    \label{eq:bernouilli_ODE}
    \left\{
    \begin{array}{rcccc}
      y' &=& (r-\delta(t))\cdot y - y^2 \\
      y(0) &=& y_0
    \end{array}
  \right.
\end{equation}
with $y_0>0$ and $\delta(t)$ continuous function satisfying $\delta(t)\stackrel{t \rightarrow +\infty}{\longrightarrow} 0$ and $\int_0^t|\delta(s)|\,d s\leq C$. Then:
\begin{equation}
  \label{eq:cv_y}
  y(t) \stackrel{t \rightarrow +\infty}{\longrightarrow} r.
\end{equation}
\end{lem}
\begin{proof}[lemma] Denote $f(t)=r-\delta(t)$, there is an explicit solution to the  Bernouilli differential equation \eqref{eq:bernouilli_ODE}: denoting $F$ an anti-derivative of $f$ (i.e. $F'(t) = f(t)$)
  \begin{displaymath}
    y(t) = \frac{e^{F(t)}}{C_0 + \int_0^t e^{F(s)}\,d s} = \frac{1}{C_0e^{-F(t)} + e^{-F(t)}\int_0^t e^{F(s)}\,d s}.
  \end{displaymath}
  Since $f(t)\stackrel{t \rightarrow +\infty}{\longrightarrow} r>0$, we deduce $F(t) = \int_0^t f(s)\,d s \stackrel{t \rightarrow +\infty}{\longrightarrow} +\infty$ and therefore $C_0e^{-F(t)}\stackrel{t \rightarrow +\infty}{\longrightarrow} 0$. Denote $m(t)=\int_0^t \delta(s)\,d s$ and notice that $m(t)\leq C$ since $\delta$ is integrable. We deduce using integration by parts:
  \begin{align*}
    &e^{-F(t)}\int_0^t e^{F(s)}\,d s \\
    &= e^{-r\cdot t}e^{-m(t)}\int_0^t e^{r\cdot s}e^{m(s)}\,d s\\
                                           &= e^{-r\cdot t}e^{-m(t)}\left(\left.\frac{e^{r\cdot s}}{r}e^{m(s)}\right|_0^t - \int_0^t \frac{e^{r\cdot s}}{r}\delta(s)e^{m(s)}\,d s   \right) \\
                                           &= \frac{1}{r} - \frac{e^{-r\cdot t}}{r}e^{-m(t)+m(0)}  -  e^{-r\cdot t}e^{-m(t)}\int_0^t \frac{e^{r\cdot s}}{r}\delta(s)e^{m(s)}\,d s\\
    &= \frac{1}{r}  -  A(t)  -  B(t).
  \end{align*}
  Since $m(t)$ is bounded, we deduce:
  \begin{displaymath}
    A(t) =  \frac{1}{r}e^{-r\cdot t}e^{-m(t)+m(0)} \stackrel{t \rightarrow +\infty}{\longrightarrow} 0.
  \end{displaymath}
  For the term $B(t)$, we use Lebesgue's dominated convergence theorem:
  \begin{align*}
    B(t) &= \frac{1}{r} \int_0^{+\infty} u_t(s)\,d s \\ \text{ with } u_t(s)&= e^{-r\cdot t}e^{s\cdot t}\mathds{1}_{[0,t]}(s) e^{-m(t)}e^{m(s)}\delta(s).
  \end{align*}
  The integrand satisfies for any $t>0$:
  \begin{displaymath}
    |u_t(s)| \leq  1 e^{2C}|\delta(s)| \quad \text{integrable}.
  \end{displaymath}
  Moreover, for any fix $s>0$, we have $\lim_{t\rightarrow+\infty} u_t(s) = 0$. Therefore we conclude that $B(t) \stackrel{t \rightarrow +\infty}{\longrightarrow} 0$. We deduce that:
  \begin{displaymath}
    y(t) = \frac{1}{C_0e^{-F(t)} + e^{-F(t)}\int_0^t e^{F(s)}\,d s} \stackrel{t \rightarrow +\infty}{\longrightarrow} \frac{1}{0+\frac{1}{r}}.
  \end{displaymath}
\end{proof}

\subsection{Extension of Fixed Graph in Higher Dimension}

Here let us redefine Equation \ref{eq:fixed_graph_Doppelganger} but extend to higher dimension $d\in \mathbb{N}$ such that we have the following:

\begin{equation}
\label{eq:higher_dim_Doppelganger}
    (s^{k})' = As^{k} - ||s|| s^{k}
\end{equation}
 where $k: 1, 2,\hdots , d$ and $||s|| = \sqrt{\frac{1}{N}\sum_{i,j}|s_{i}^{j}|^{2}}$ \\
\begin{theorem}
\label{thm:Doppelganger_fixed_graph_higher_dim}
  Consider the fixed-graph Doppelg\"anger model in dimension $d$ \ref{eq:higher_dim_Doppelganger}. Assume that $A$ has a dominant eigenvalue $r$ and the initial condition $s^{k}(t=0)$ for any $k\in[1,d]$ has a non-zero component along the associated Perron-Frobenius eigenvector. Then the solution $s^{k}(t)$ converges (exponentially fast) to a Perron-Frobenius eigenvector $\vect{u}_r$ of $A$.
\end{theorem}

\begin{proof}
     We keep the composition of $\Rn$ as the union of eigenspaces with the dominant eigenspace denoted as $E^{r}$. The solution to the Doppelgänger model retains the decomposition:
    \begin{equation*}
        s(t) = \alpha(t)u_{t} +v(t)
    \end{equation*}
    although now that we extend to dimension d, we rewrite the equation using the following:
    \begin{equation}
        \label{eq:fixed_sol_in_d}
        s^{k}(t) = \alpha^{k}(t)u_{r}  + v^{k}(t)
    \end{equation}
    where k is defined similarly in Equation \ref{eq:higher_dim_Doppelganger}. Substituting equation \ref{eq:fixed_sol_in_d} into the model equations.
    \begin{align*}
        (\alpha^{k})'(t) u_{r}+v'(t) &= A\big[\alpha^{k}(t)u_{r}+v^{k}(t)\big]-|\lvert s \rvert| \big[\alpha (t)u_{r}+v(t)\big] \\
        &\implies \left\{
         \begin{array}{cc}
             (\alpha^{k})' u_{r} &= r\alpha^{k}u_{r} - ||s|| \alpha^{k}u_{r} \\
             (v^{k})' &= A_{|E_{s}}v^{k} - ||s|| v^{k}
         \end{array}
         \right.\\
    \end{align*}

We maintain that the solution is bounded by Proposition \ref{prop:dopp_bounded} since we retain the same $\phi$ (a bounded function) for the fixed Doppelg\"anger in higher dimension. The proof follows similarly with adjustments such that:
\begin{align*}
      \alpha ^{k}(t) &= \exp\left(r t - \int_0^t μ(\textbf{s}(z))\, dz\right)\alpha ^{k}(0)\\
      \alpha^{k}(t) &\leq C \\
      rt - C &\leq \int_{0}^{t} ||s(z)||dz
\end{align*}
Moreover, the solution for $\norm{v^{k}(t)} \stackrel{t \rightarrow +\infty}{\longrightarrow} 0$ exponentially by using the conditions above. Looking at convergence in time we look into the convergence of $\alpha^{k}(t)$
\begin{align*}
    \lim_{t\rightarrow \infty}||s(t)|| &= \lim_{t\rightarrow \infty}\bigg[\norm{\alpha^{k}(t)u_{r}}+\delta(t)\bigg]\\
    \implies |\delta(t)| &= \bigg| ||\alpha^{k}(t)u_{r} + v^{k}(t)|| - ||\alpha^{k}(t)u_{r}||\bigg|\\
    \implies |\delta(t)|&\leq v^{k}(t) \rightarrow 0
\end{align*}
By utilizing the normalization of eigenvector $u_{r}$ it holds that $||s(t)|| = \alpha^{k}(t)+\delta(t)$ where by substituting into equation \ref{eq:higher_dim_Doppelganger} we get $$(\alpha^{k})' = (r-\delta(t))\alpha^{k}-(\alpha^{k})^2$$
which is has an explicit solution and converges to r, by Lemma \ref{lem:Bernouilli_ODE}, when $\delta(t)$ tends to zero.
\end{proof}

% Bibliography
\newpage
\bibliography{doppelganger}
\bibliographystyle{plain}

\end{document}